\documentclass[10pt, english,pra,aps,twocolumn,floatfix,superscriptaddress]{revtex4-2}
\usepackage[T1]{fontenc}
\usepackage[latin9]{inputenc}
\usepackage{color}
\usepackage{babel}
\usepackage{verbatim}
\usepackage{refstyle}
\usepackage{amsthm}
\usepackage{amsmath}
\usepackage{graphicx}
\usepackage{algorithmicx}
\usepackage[ruled,linesnumbered]{algorithm2e}
\usepackage{color}
\usepackage[pdfstartview={FitH},bookmarks=false]
 {hyperref}

\theoremstyle{plain}
\newtheorem{thm}{\protect\theoremname}
  \theoremstyle{remark}

\providecommand{\remarkname}{Remark}
\providecommand{\theoremname}{Theorem}
\newtheorem{theorem}{Theorem}[section]

\newtheorem{lemma}[theorem]{Lemma}

\makeatletter

\newcommand{\ket}[1]{\left| #1\right\rangle}        % ket vector
\newcommand{\bra}[1]{\left\langle #1\right|}        % bra vector

\newcommand{\dimHB}{d}%\mathcal{D}_{A} }

\newcommand{\weigt}{\omega}

\newcommand{\diagS}{\alpha}
\newcommand{\diagHB}{\beta}

\newcommand{\Sys}{T}
\newcommand{\A}{A}

\newcommand{\HBAC}{Heat-Bath Algorithmic Cooling$\,$}
\begin{document}

\title{No-go Theorem of Purification}

\author{Sadegh Raeisi}\email[]{sraeisi@sharif.edu}
\affiliation{Department of Physics, Sharif University of Technology, Tehran, Iran}

\begin{abstract}
The Shannon's bound for compression is one of the key restrictions for the compression of quantum information. Here we show that the unitarity of the compression operation imposes new bounds on the compression that are more limiting than Shannon's compression bound. 
This translates to a no-go theorem for the purification of quantum states. For a specific case of a two-qubit system, our results indicate that it is not possible to distill purity beyond the maximum of the individual purities. 
We show that this restriction results in the cooling limit of the heat-bath algorithmic cooling techniques.
We formalize the limitations imposed by the unitarity of the compression operation in two theorems and use the theorems to show that the limitations of unitarity lead to the cooling limit of heat-bath algorithmic cooling. 
To this end, we introduce a new optimal cooling technique and show that without the limitations of the unitary operations, the new cooling technique would have exceeded the limit of Heat-bath algorithmic cooling. 
This work opens up new avenues to understanding the limits of dynamic cooling. 
\end{abstract}
\maketitle

Quantum mechanics predicts some peculiar and fascinating phenomena
that have been demonstrated with sophisticated experiments. 
Some of these quantum effects have even been utilized for 
applications and technologies  
 such as quantum computing and quantum sensing. 
However, realizations of these quantum experiments and technologies are challenging, mainly because quantum effects are often fragile and could be sensitive to imperfections of the implementation  \cite{wang2013precision}.

Techniques like quantum error correction
and fault tolerance  \cite{nielsen2002quantum} were 
invented to counteract some of these 
imperfections. These solutions usually focus only on the implementation of operations 
and not the state preparation.  
Also, they often rely on large supplies of 
high-quality quantum states.  
Yet, one of the main challenges is the inability to 
prepare high-quality quantum states. These states
are referred to as ``pure'' states and 
make one of the key ingredients 
of many quantum experiments and applications \cite{ben2013quantum}. 
For instance, the vast majority of quantum algorithms such as 
Shor's algorithm \cite{shor1994algorithms} or 
even quantum error correction and fault 
tolerance techniques require 
pure ancillary qubits \cite{nielsen2002quantum}. 
This is one of the reasons why quantum systems are cooled.

%%%%HBAC

\HBAC (HBAC) provides an alternative solution to this problem. 
These techniques use auxiliary degrees of freedom in the 
system to compress and transfer the entropy away from a 
target subsystem which is of interest to the quantum experiment. 

\begin{figure}
	\begin{centering}
	\includegraphics[width=.7\columnwidth]{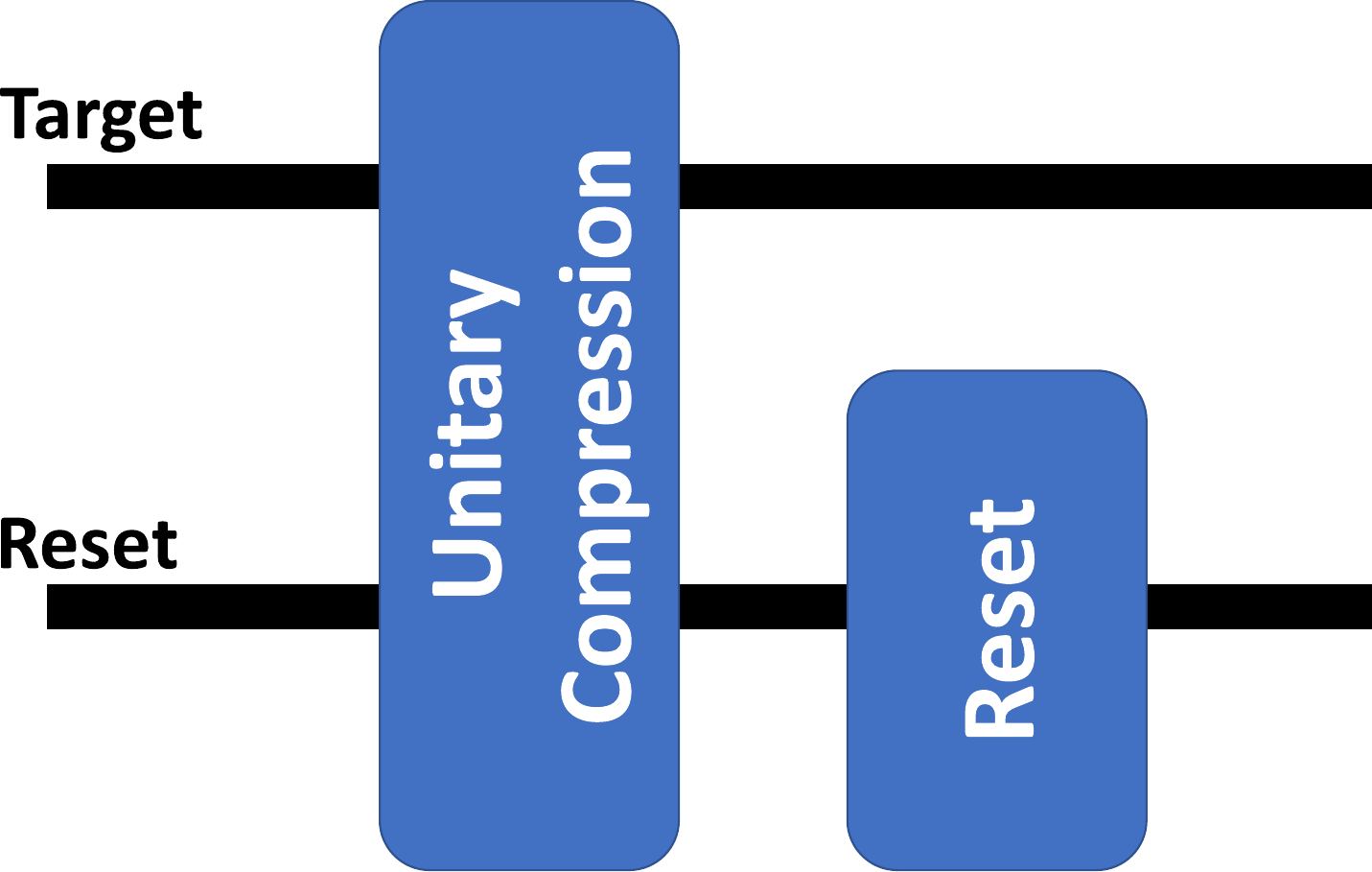}
	\par\end{centering}	
	\caption{\label{fig:Schematic-HBAC}
		Schematic Picture of HBAC. 
		The compression operation 
		compresses the entropy of the combined system 
		and transfers that to the reset element. 
		The compression operation increases the entropy 
		of the reset element. 
		The reset operation refreshes the reset element 
		and brings it back to its equilibrium or initial state which has lower entropy. 		
		}
\end{figure}

This idea was first proposed by Schulmann and Vazirani \cite{schulman_scalable_1998} for closed quantum systems. 
Later Boykin et al. extended this idea to open systems \cite{boykin_algorithmic_2002}. 
In particular, for a system where one of the auxiliary elements has a strong interaction with the heat-bath, they proposed to use the interaction to extract the entropy out of the system and into the heat-bath. 
This allows to cool beyond Shannon's bound for compression.

HBAC is particularly useful for spin-based quantum systems.  Often the thermal state of spin systems is highly mixed which is partly due to their small energy gaps. 
For instance, for Hydrogen nuclear spins in a 10-Tesla field and at room temperature, the Boltzmann distribution gives a polarization of  ($~10^{-5}$). This is only slightly more pure than a maximally mixed state.   
HBAC has recently been used for the enhancement of the polarization of the Nitrogen nuclear spins in Nitrogen-Vacancy(NV) centers in diamond \cite{zaiser2021experimental}. In this experiment, two carbon spins were used as the auxiliary elements to absorb the entropy from the Nitrogen spin through the compression operation of HBAC. This experiment was implemented at room temperature.

In \cite{schulman_physical_2005}, Schulman et al. introduced the optimal HBAC technique which they called the "Partner Pairing Algorithm (PPA)". However, they found that even PPA cannot always  converge to a fully pure state and there is still a limit. 
This was counter-intuitive. One would expect that in an open system setting, it should be possible to iteratively compress and transfer  entropy away from the target element to the auxiliary element and then to the bath repeatedly until all the entropy is extracted from the target element and it converges to a pure state.

Although \cite{schulman_physical_2005} showed that HBAC has a cooling limit, they did not find the limit. This was because PPA has a complex process and they could not find the asymptotic state of PPA. While the existence of the limit was verified numerically, the asymptotic state remained unknown for close to a decade. 

In 2015, Raeisi and Mosca solved this problem and established the asymptotic state of PPA \cite{raeisi2015asymptotic}. 
Using the asymptotic state, they also found the cooling limit of PPA. Since PPA is an optimal HBAC technique, the limit applies to all HBAC techniques.  More specifically, \cite{raeisi2015asymptotic} established the cooling limit of all HBAC techniques. 

Despite the establishment of the limit, the roots of the limit remained unclear. Note that Raeisi and Mosca found the asymptotic state by analyzing and proving some bounds for the iterations of PPA \cite{raeisi2015asymptotic}. While these bounds helped to establish the asymptotic state, they did not clarify why the overall cooling of the process is bounded. 

Note that this problem is different from the third law of thermodynamics. A similar paradigm to HBAC has been investigated for the derivation of the third law for quantum thermodynamics \cite{browne2014guaranteed, reeb2014improved, masanes2017general}. However, the research in this field is motivated by different sets of assumptions. The obvious distinction is that the focus of the research on the third law is on processes that converge to zero temperature. However, for the HBAC limit, the process is not always converging to the zero temperature, i.e. a completely pure state.

Here, we trace back the cause of the limit to the unitarity of the compressions operations. 
We show that the unitarity of operations limits the compression.  For instance, we show that with two qubits, it is not possible to compress and transfer the entropy of one of the qubits to the other one. As a result, for a two-qubit system, we prove that it is not possible to increase the purity of the qubits beyond individual purities. 
We formalize these limitations in terms of bounds 
on the purity after the unitary compression operation. We then show that these bounds lead to the HBAC limit. We also show that if the compression could exceed the limit, it would have been possible to cool beyond the cooling limit of HBAC. To this end, We introduce a new HBAC technique that helps see the effects of the compression bound more clearly.

The structure of this paper is as follows. 
We first describe our notation. 
Next, we proceed to introduce the new HBAC technique. 
We then prove two limitations of unitary compression operations. 
Next, we show that without these limitations, 
the new HBAC technique would exceed the cooling limit of HBAC and that because of these limitations, it is bounded by the cooling limit of HBAC. 
This concludes our result and shows that the HBAC limit is imposed by the limitations of the unitarity of the compression operations.

%%%Notation
\subsubsection*{Notation}
For HBAC, often the system is divided into two subsystems.  
The target subsystem is referred to as the ``computation element'' and the subsystem that is in strong interaction with heat-bath 
is called the ``reset element''. 
Often the system is assumed to 
be a combination of qubits in which case, these would be referred to as the computation and reset qubits. 
For most of this work, we assume that the system comprises $ n $ 
qubits,  $n-1$ computation qubits and one reset. 
We index the qubits from $ 1 $ to $ n $ and refer to them as $ Q_j $ for the $ j $th qubit. 
Note that, although all the computation qubits are purified to some extent, some are more purified than the others and in some sense, some computation qubits act as auxiliary elements for the other computation qubits.

The state of the target and reset elements are 
described by non-negative density matrices in Hilbert spaces $ \mathcal{H}_T $ and $ \mathcal{H}_R $ respectively. 
We use $ \lambda\left( \rho\right)  $ to refer to the eigenvalues of the density matrix $ \rho $,  
$ \bar{\lambda} \left( \rho\right)$ to represent the sorted array of $ \lambda\left( \rho\right)  $  
in decreasing order and
$ \lambda_i $ for the $ i $th element of 
$ \bar{\lambda}\left( \rho \right)  $. 
Also, we use subscripts 
$ T $ and $ R $ to refer to the target and reset 
elements respectively.

HBAC comprises two main operations. 
First is the compression operations which compresses and transfers the entropy away 
from the target(computation) elements. 
This is done by a unitary operation $U$. 
Mathematically the compression is 
\begin{equation}\label{eq:HBAC_compression}
\rho_{T,R}\rightarrow U \rho_{T,R} U^{\dagger},  
\end{equation}
The compression operation increases the entropy of 
the reset element. So, it is followed by the ``reset'', 
where the reset element loses its accumulated entropy 
to the heat-bath. Often, this is a relaxation process 
that takes the reset element to its equilibrium state. 
This state is referred to as the ``reset state, $\rho_R$''.
For a single qubit  reset element, the reset state would be
\begin{equation}\label{eq:rho_e}
\rho_R(\epsilon_R) = \frac{1}{e^{\epsilon_R}+e^{-\epsilon_R}}
\left(\begin{array}{cc}
e^{\epsilon_R} & 0\\
0 & e^{-\epsilon_R}
\end{array}\right), 
\end{equation}
where $\epsilon_R$ is referred to as the ``polarization'' of the reset state and determines the purity of the reset state. 
For a thermal reset, this is set by the Boltzmann distribution and depends 
on the temperature and energy gap between the two qubit states. 

In general, any qubit state in its diagonal basis can be written in this form, and the polarization can be defined as 
\begin{equation} \label{eq:polarization}
\epsilon\left( \rho\right)  = \frac{1}{2}\log(\frac{\bar{\lambda}_1}{\bar{\lambda}_{2} }),   
\end{equation}
where $\bar{\lambda}_1$ and $\bar{\lambda}_2$ are the eigenvalues of the state $\rho$ in decreasing order. Polarization can be used as a measure of purity.

For a maximally mixed qubit state $ \epsilon=0 $ and for a pure qubit state $ \epsilon\rightarrow\infty $.

The definition of the polarization  can be 
extended to a $ d $-level system as
\begin{equation} \label{eq:polarization-qudit}
\varepsilon\left( \rho\right)  = \frac{1}{2}\log(\frac{\bar{\lambda}_1}{\bar{\lambda}_{\dimHB} }),   
\end{equation}
where $d$ is the dimension of the Hilbert space % of the reset element 
and $\lambda_1$ and $\lambda_d$ are the smallest and largest eigenvalues of the density matrix. 
Note that for $ d > 2 $ this 
is not necessarily a good purity measure and it is in general more challenging to define a measure of 
purity for 	qudits \cite{gour2015resource}. 

Alternatively, entropy can be used as a measure of purity. 
The entropy can be expressed as 
\begin{equation}
S(\rho) = -\sum_i \lambda_i \log(\lambda_i).
\end{equation}
It is easy to show that for a qubit, the entropy can be expressed in terms of the polarization as $S(\rho) = -(2 e^{2 \epsilon} \epsilon)/(1 + e^{2 \epsilon}) + log(1 + e^{2 \epsilon})$. This is a strictly decreasing function which means that increasing the polarization is equivalent to decreasing the entropy. 
In this paper, we sometimes refer to the cooling process as increasing the polarization or decreasing the entropy.

The combination of the compression and the reset
can be mathematically described as 

\begin{equation}\label{eq:dynamic_cooling}
%\CoolingC \left( \rho_{T,R}\right) = 
\rho_{T,R}^{out} = 
\text{Tr}_{R}\left(U \rho_{T,R} U^{\dagger}\right)\otimes \rho_R(\epsilon_R),  
\end{equation}
where $ U $ is a unitary compression operation and  
$ \text{Tr}_{R} $ is the partial trace taken over the 
reset element. 
Figure (\ref{fig:Schematic-HBAC}) gives 
a schematic picture of the process. 

The Partner Pairing Algorithm (PPA) strategy that was introduced in \cite{schulman_physical_2005} uses sorting for the compression operation and they showed that it is the optimal compression operation. 
More specifically, in each iteration, PPA sorts 
the diagonal elements of the density matrix in decreasing order. Schulman et al. also showed that this process does not always converge to a pure state. However, due to the complexity of the sort operations that are constantly changing, the asymptotic state remained unknown for close to a decade. 
In 2015, Raeisi and Mosca \cite{raeisi2015asymptotic} solved this problem and showed that the PPA process converges to the following  diagonal state %with the following diagonal elements 
\begin{equation}
    \text{diag}(p_0 \{1, e^{-2 \epsilon_R}, e^{-4 \epsilon_R}, \dots, e^{-2 (2^{(n-1)} - 1) \epsilon_R}  \}) \otimes \rho_R(\epsilon_R),
\label{Eq:rho_HBAC}
\end{equation}
where $\epsilon_R$ is the reset polarization, $ p_0 $ is the normalization factor, and diag represents a diagonal matrix of its inputs. They showed this result for any value of the reset polarization $\epsilon_R$ as well as any number  of computation and reset qubits. They also extended the results for general multi-level systems. 
Using the asymptotic state, they established the following asymptotic bound for the polarization of the target qubit
\begin{equation}
\epsilon_T^{\infty} = 2^{(n-2)} \epsilon_R.
\end{equation}

%%%Notation

Next we introduce a new HBAC algorithm which we refer to as the ``recursive HBAC technique''. This algorithm can  in principle converge to the HBAC limit, although it may not be optimal in terms of the number of steps required to get to the vicinity of the asymptotic state. However, it makes a nice toy HBAC algorithm to show why the cooling process cannot go beyond the cooling limit.

\subsubsection*{Recursive HBAC algorithm}
%%%New HBAC algorithm

\begin{algorithm}[t]
\SetArgSty{textnormal}
\SetKwInOut{Input}{Input}
\SetKwInOut{Output}{Output}
\DontPrintSemicolon
\Input{Number of qubits $ n $, Polarization of the 
		reset qubit $ \epsilon_0 $, $\delta $}
	\Output{Final state after recursive purification }
	\DontPrintSemicolon
	\SetKwFunction{FMain}{Main}
	\SetKwProg{Fn}{RecursivePurify}{:}{}
	\Fn{\FMain{$n$, $\epsilon_{0}$, $\delta$}}{
%\Begin{
		$ \rho_{\text{T,A}} = \rho_{\epsilon_0}^{\otimes n} $\;
		\eIf{n=2}{
			\KwRet$ \rho_{\epsilon_0}^{\otimes 2} $\;
			
		}{
			\While{ $ \| \rho_{T, A}-\rho^{\infty}\left( n,\epsilon_0\right) 
				\|_1 \geq \delta $}{
				$ \rho_{T, A} = sort( \rho_{T, A} ) $\;
				$ \rho_{T} = \text{Tr}_{A}\left(\rho_{T, A}\right)   $\;
				$ \rho_{T, A} = \rho_{T} \otimes \text{RecursivePurify}\left( n-1,\epsilon_0, \delta\right) $\;
				\KwRet $ \rho_{T, A}$ \;
			}
		}
	}
\caption{Recursive HBAC} \label{recsv_HBAC}	
\end{algorithm}

We take $ Q_1 $ to be the reset qubit with 
reset polarization of $ \epsilon_{R}$ and the rest of the qubits to be the computation qubits. 
We also take $ Q_n $ to be the main target. 
This means that although the polarization of all the computation qubits increases, we focus on $ Q_n $. 
We use $ \epsilon_j $ to refer to the polarization 
of qubit $ Q_j $.

Here we introduce a recursive algorithm for 
HBAC. 
The main idea is that we increase the purity of qubits one at a time and 
then use the combination of the cooled qubits 
as the reset for the next qubit. 
More specifically, the recursive algorithm starts with the 
first two qubits, $ Q_1 $ and $ Q_2 $. 
For the first part, we swap the state of 
$ Q_1 $ and $ Q_2 $ which polarizes $ Q_2 $ 
to  $ \epsilon_{R} $. 
We then reset $ Q_1 $.  
Now we use the combination of $ Q_1 $ and $ Q_2 $ 
as the reset element for the next qubit. 
We apply one round of purification. 
This would reduce the purity of $ Q_1 $ and $ Q_2 $. 
So they should be re-purified which involves 
resetting $ Q_1 $ and cooling $ Q_2 $ again. 
After that, the purification of $ Q_3 $  with re-purified 
$ Q_1 $ and $ Q_2 $ as the reset element is repeated. 
We keep doing this until $ Q_3 $ converges to 
its limit % of  $ 2\epsilon_{0} $ 
i.e.  stops changing within some threshold $ \delta $. 
Now we can move to the next qubit.  
Similarly, all the qubits can be purified this way. 

Since the algorithm only asymptotically gets to 
the limit, we need to add a parameter, $ \delta $ to specify how close we want to get to the asymptotic 
state. We refer to the asymptotic state of 
HBAC as $ \rho^{\infty}\left( n,\epsilon_0\right)  $. 
We also assume that $ n>1 $ and use $ \| A \|_1 $
for the norm one of an operator. 
The following pseudo-code in Algorithm (\ref{recsv_HBAC}) gives a more concrete description of the algorithm. 

%Pseudocode

This is a recursive algorithm because at each iteration 
of the algorithm, all the 
previous qubits should be recursively purified.

The intuition behind the algorithm is that, for each qubit, we want to first increase the purity of the reset element as much as possible and then use the purified reset element to cool the target qubit. But the reset element itself contains multiple qubits. So we break the reset element into one target qubit and a reset element. We then repeat this for the new reset element. 
Figure (\ref{fig:HBAC_limit}) depicts this intuition.

\begin{figure}[t]
	\begin{centering}
		\includegraphics[width=\columnwidth]{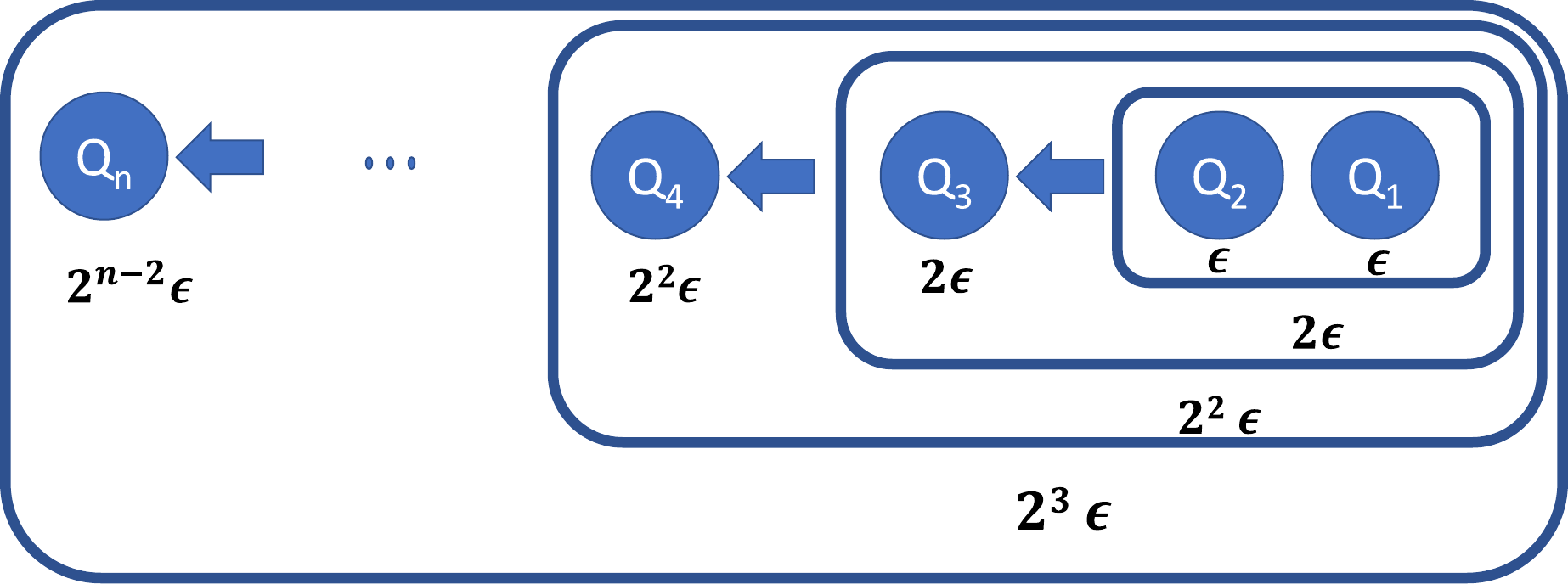}
		\par\end{centering}
	
	\caption{\label{fig:HBAC_limit}
		Schematic depiction 
		of the recursive HBAC technique.   
		The circles and boxes represent target qubits 
		and reset elements at each stage. 
		As we move to the left, 
		target qubits become part of the reset element
		for the next target qubit. 
		This figure also shows  how the bound 
		for the cooling in this algorithm is derived. 
		Each box (except for the most inner one) 
		includes a qubit, which is the target, 
		and another box which serves as the auxiliary 
		element. From theorem  (\ref{thm:cooling_power}), 
		from each qubit to the next one, the polarization 
		bound of the reset element doubles which 
		gives the exponential growth of the bound. 
	}
\end{figure}
This algorithm, like any other HBAC algorithm, is bounded by the cooling limit of HBAC established in \cite{raeisi2015asymptotic}.
To find and understand the bound for this specific algorithm, we prove two theorems about the compression unitary operations.

%%%New HBAC algorithm

\subsubsection*{Limitations of Unitary Compression}

%%% Thm for Unitary operations
We first investigate the optimal compression unitary, 
i.e. any operation that maximizes the purity 
of the output reduced density matrix of the target 
element, i.e. $ \rho_{T}^{out} = \text{Tr}_{R}\left(U \rho_{T,R} U^{\dagger}\right)$. 

For any unitary compression operation, 
there is a class of unitary operations that give the
same purity for the output target state. This class is generated by multiplying the compression unitary with local unitary operations. 
So to find the cooling 
bounds, without loss of generality, we can
focus on channels which preserve the diagonal 
basis of the target element and 
keep the output state $ \rho^{out}_T $ diagonal.
In other words, for any compression unitary that gives a non-diagonal density matrix
for the target element, it is always possible to combine it with a 
local unitary on the target that keeps $ \rho_T $ diagonal and  
leads to the same purity for the target element.

To get to our main result, we use the following lemma which indicates that the 
optimal compression unitary
can be chosen to be a permutation. 

\begin{lemma}
	Assume that we are given a target and an auxiliary system with 
	dimensions 2 and $\dimHB$ respectively. 
	Also assume that their initial states are given by
	$\bar{\lambda}\left( \rho^T \right) = \{\diagS, 1-\diagS\}$ and 
	$\bar{\lambda}\left( \rho^A \right) = \{\diagHB_1, \diagHB_2, \cdots,\diagHB_{\dimHB}\}$. 
	Given the quantum channel in equation (\ref{eq:dynamic_cooling}), 
	the optimal compression operation can always be chosen to be a permutation. 
\end{lemma}

\begin{proof}
	Without loss of generality, we can assume 
	that the initial state of the full system is diagonal. 
	Also as explained above, the optimal compression operator 
	can be chosen such that it
	keeps the reduced density matrix of the target diagonal. 

	We write the compression unitary $U$ as
	\[U = \sum_{i,j} u_{i,j} \ket{i}_{T,A}\bra{j} \] 
	with $i$ and $j \in \{1, 2,  \cdots, 2\dimHB\}$, enumerating
	over the full basis of the target and auxiliary elements.

	For the first $\dimHB$ elements (first block), 
	$\ket{i}_{\Sys,\A} = \ket{ 0 }_{\Sys} \ket{ i }_{\A}$ and for 
	the second $\dimHB$ elements (second block), 
	$\ket{i}_{\Sys,\A} = \ket{ 1 }_{\Sys} \ket{ i-d }_{\A}$. 
	The outcome of the compression can be calculated as
	\[ \rho^{out}_{T,A}  = \sum_{i,j} \left(\sum_{z}{ u_{i,z}
		\lambda_z\left( \rho_{T,A}\right)  u^*_{j,z}}\right)  \ket{i}\bra{j}. \]
	
	Since the output state of the target element is diagonal, 
	we can focus on the diagonal density 
	matrices for calculations of the purity. 
	The diagonal elements are given by $i=j$, i.e.  
	\[ \left( \rho^{out}_{T,A}\right)_{i,i}  =  \left(\sum_{z}^{2\dimHB}{ \mid u_{i,z}\mid^2
		\lambda_z\left( \rho_{T,A}\right)  }\right). \]
	For the output target density matrix, 
	we get
	\begin{equation}\label{eq:output_purity}
	\alpha^{out} = \sum_{i=1}^{\dimHB}\sum_{z=1}^{2\dimHB}{\mid u_{i,z}\mid^2 \lambda_z\left( \rho^{out}_{T,A} \right) }=\sum_{z=1}^{2\dimHB}{\weigt_z \lambda_z\left( \rho^{out}_{T,A} \right) },
	\end{equation}	
	where 
	\[\weigt_z =  \sum_{i=1}^{\dimHB} \mid u_{i,z}\mid^2 \]
	defines some weights. 
	Note that the sum goes over the first 
	$\dimHB$ elements. 
	To maximize the purity $ \epsilon\left( \rho^{out}_{T,A} \right)  $, the compression unitary or more 
	specifically the weights  $ \weigt_z $ should be 
	set such that the $ \alpha^{out} $ in 
	equation (\ref{eq:output_purity}) is maximized. 
	
	The unitarity of $ U $ implies that 
	$0\leq \weigt_z\leq 1$. 
	Also the sum over all $2\dimHB$ weights would 
	be  
	\[\sum_{z=1}^{2\dimHB}\weigt_z = \sum_{i=1}^{\dimHB} \sum_{z=1}^{2\dimHB} \mid u_{i,z}\mid^2=\sum_{i=1}^{\dimHB} 1 =\dimHB. \]
	This has an important implication. The optimal 
	compression operation $U$ should be chosen 
	such that  in equation (\ref{eq:output_purity}),
	the larger elements get 
	the largest possible weights. 
	Considering $ \weigt_z\leq 1 $ and that 
	$\sum_{z=1}^{2\dimHB}\weigt_z =\dimHB$, 
	the weights of the first $ d $ elements of  $\bar{\lambda}\left( \rho_{T,A} \right)$ should be one and the rest should be zero. 
	This indicates that the matrix elements of the optimal operation are
	either zero or one. Note that any unitary operations that its matrix elements are only zero or one, can only have a single non-zero element in each row or column because it has to preserve the norm. 
	This means that the operation would be a permutation matrix. 
\end{proof}

This lemma indicates that, there exists an optimal compression
such that starting from a diagonal density matrix, 
the compression operation keeps the density matrix diagonal and only the order of 
diagonal elements would change. 
Note that this permutation is not unique. 
One clear choice is the sort operations, i.e. the operation that sorts  the eigenvalues of the density matrix, $ \lambda_i $s. 
We will use the sort operation for the rest of this paper. 

The following two theorems establish two limitations 
imposed by unitarity of the compression operations. 
Note that the compression operation acts between two subsystems, one of which is the target element. 
The other element could be the reset qubit or a combination of reset and some of the computation qubits (as in the recursive HBAC technique). 
For simplicity, we refer to the second 
subsystem as the auxiliary subsystem and use subscript 
$A$ to refer to it. 

Now we proceed to the no-go theorem for the two qubit purification setting. 

\begin{thm}\label{thm:no-go}
	The two-qubit purification no-go theorem: 
	Assume that we are given two qubits 
	for the target and the auxiliary systems with 
	$\bar{\lambda}\left( \rho_{\Sys} \right) = \{\alpha, 1-\alpha\}$ and 
	$\bar{\lambda}\left( \rho_{\A} \right) = \{\beta, 1-\beta\}$ and 
	that the state of the target element  
	after purification is given by 
	$\bar{\lambda}\left( \rho_{\Sys}^{out} \right) = 
	\{\alpha^{out}, 1-\alpha^{out}\}$. Then
	\begin{equation}\label{eq:no-go-qubit}
	\alpha^{out} \leq \max(\alpha, \beta).
	\end{equation} 
\end{thm}

\begin{proof}
	Without loss of generality, we assume that 
	the initial state is diagonal. We have 
	\[\lambda\left( \rho_{T,A}\right) =   \{\alpha\beta,\alpha(1-\beta), (1-\alpha)\beta, 
	(1-\alpha)(1-\beta)\}.\]
	Considering that the optimal unitary operator 
	should place the larger elements of $ \lambda \left( \rho_{T,A}\right) $
	in the first block, 
	there are two possibilities, if 
	\[\alpha(1-\beta)\geq (1-\alpha)\beta,\]% the 
	 $ \lambda \left( \rho_{T,A}\right) $ is already sorted and 
	we get  $\alpha^{out} = \alpha((1-\beta)+\beta)=\alpha$.
	And if 
	$ \alpha(1-\beta) < (1-\alpha)\beta $, then
	$\alpha^{out} = \beta((1-\alpha)+\alpha)=\beta$. So 
	the optimal purification gives 
	$ \alpha^{out} = \max\left( \alpha, \beta\right)  $. 
\end{proof}

This result can be used to bound the purity after the compression. 
Note that for two states $\rho$ and $\rho'$, if
$\bar{\lambda_1(\rho)} > \bar{\lambda_1(\rho')}$, then it for the second eigenvalue we have
%implies that 
\[\bar{\lambda_2(\rho)} = 
1 - \bar{\lambda_1(\rho)} \leq 
1 - \bar{\lambda_1(\rho')} = \bar{\lambda_2(\rho)}. \]
This indicates that $\epsilon(\rho) \leq \epsilon(\rho') $. 
Similarly, the condition in equation (\ref{eq:no-go-qubit}) can be used to derive the following bound on the purity after the compression 
\begin{equation}
\epsilon \left( \rho^{out}_{T}\right) \leq \max \left( \epsilon\left( \rho_{T}\right),\epsilon\left( \rho_{A}\right)\right),  
\end{equation}
i.e. the polarization of the output target qubit is bounded by the maximum 
of the initial polarization of the target and auxiliary elements. 

This theorem has a significant implication. 
It shows that for two qubits, there is no unitary operation 
that can compress and transfer entropy from the target qubit beyond the maximum of the polarization of the two qubits. 
For instance, assume that initially, the two qubits have polarization $\epsilon_1$ and $\epsilon_2\geq \epsilon_1$. Now assume that the first qubit is the target. This theorem shows that at best, one can swap the two qubits, and it is not possible to compress the entropy. We will show that this no-go theorem limits the cooling of HBAC.

The result of Theorem (\ref{thm:no-go}) can be generalized for cases where the auxiliary element is not a qubit and belongs to a $d-$dimensional Hilbert space with $d > 2$. This is the content of our next theorem. 

For the general case where $\dimHB\geq 2$, the 
eigenvalue string is of the form 
\begin{align*} % the "starred" equation environments produce no equation numbers
\lambda \left( \rho_{T, A}\right) =\hspace{5em} & \\\{
\underset{\text{Block 1}}
{\underbrace{\diagS\times \{\diagHB_1, \diagHB_2, \cdots, \diagHB_{\dimHB} \}}}  % if no alignment is needed, use the gather* instead of the align* env.
&, \underset{\text{Block 2}}
{\underbrace{(1-\diagS)\times \{\diagHB_1, \diagHB_2, \cdots,\diagHB_{\dimHB}\}}}\}.
\end{align*}
This array is not necessarily sorted, i.e. there could 
exist indices $i$ and $j$ such that 
$\diagS \diagHB_i < (1-\diagS) \diagHB_j$.
The optimal compression operation would switch the orders 
of these elements and make sure that the first block 
contains the largest ones. 
More precisely, 
the optimal compression operation would 
replace the $m$ last elements of the first block 
with the $m$ first elements of the second block, with 
$m$ some integer that is less than $ d/2 $ and 
that depends on the order of the array.
We refer to these as ``crossing''s. 
Figure (\ref{fig:Schematic-Crossing}) gives a schematic 
description of the crossings. 
\begin{figure}
	\begin{centering}
		\includegraphics[width=\columnwidth]{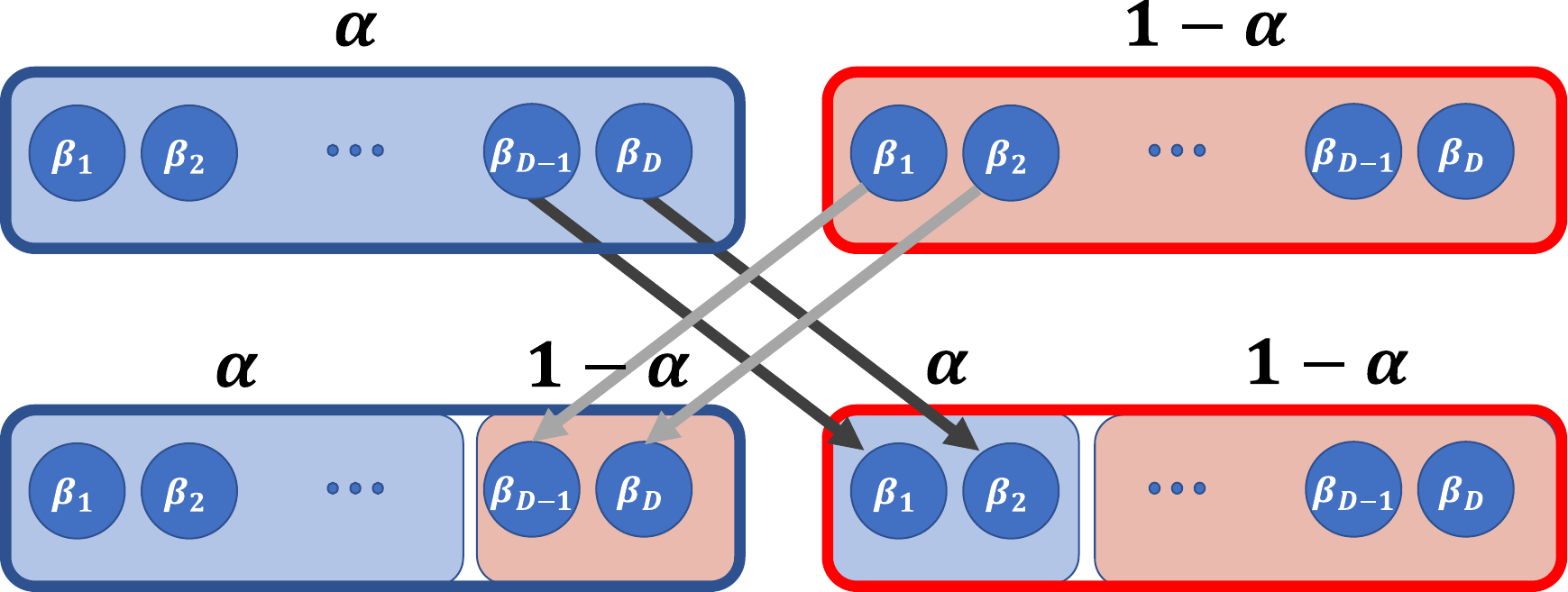}
		\par\end{centering}
	
	\caption{\label{fig:Schematic-Crossing}
		Schematic picture of crossings. 
		The top panel is the array of the eigenvalues 
		before the sort operation and 
		the one below is after the sort is applied.
		The two boxes represent the first and the 
		second block of the diagonal elements of the 
		density matrix $ \lambda\left( \rho_{T,A}\right)  $. 
		The elements in the first box are the first half 
		of the eigenvalues. 
		They can be calculated from multiplication of 
		first eigenvalue of the target state, 
		i.e. $ \alpha $ by  the eigenvalues of 
		the reset element, i.e. $\{\beta_i\}$. 
		Similarly, the second box is the second half the 
		eigenvalues of the density matrix which are
		 calculated from multiplying  
		$\{\beta_i\}$    by $ 1-\alpha$. 
		Some of the elements from the second 
		block are larger 
		and the optimal operation places them 
		in the first block to increase the 
		probability of the state $\ket{0}$ for the target 
		qubit (lower panel). 
		These crossings represent a key part of the purification dynamics. 
	}
\end{figure}

We break $ \bar{\lambda}\left( \rho_A \right)  $ in to three groups and for simplicity, we introduce the following parameters
\begin{align}
\delta_1 &= (\beta_1 + \cdots+\beta_{m}) \cr
\delta_2 &= (\beta_{m+1} + \cdots+\beta_{d-m}) \cr
\delta_3 &= (\beta_{d-m+1} + \cdots+\beta_{d}). 
\end{align} 
Note that $ \delta_1+\delta_2+\delta_3 = 1 $.

After the optimal compression operation,  
the first element of $ \rho_{T} $ changes to 
\begin{equation}\label{eq:alphaout_crossing}
\diagS^{out} = \sum_{i=1}^{\dimHB-m}{\diagS \diagHB_i}+
\sum_{i=1}^{m}{(1-\diagS) \diagHB_i}=\diagS  \delta_2+\delta_1.
\end{equation}

Now we get to the generalization of the 
theorem (\ref{thm:no-go}).

\begin{thm}\label{thm:cooling_power}
	The purification bound:\\
	Assume that we are given a target and an auxiliary system with 
	dimensions 2 and $\dimHB$ 
	respectively. Also assume that their initial states are given by
	$\bar{\lambda}\left( \rho_T \right) = \{\diagS, 1-\diagS\}$ and 
	$\bar{\lambda}\left( \rho_A \right) = \{\diagHB_1, \diagHB_2, \cdots,\diagHB_{\dimHB}\}$ and 
	that the state of the target element 
	after compression operation is given by 
	$\bar{\lambda}\left( \rho_{\Sys}^{out} \right) = 
	\{\alpha^{out}, 1-\alpha^{out}\}$. 
	Then
	\begin{equation}\label{eq:thm2}
	%\frac{\alpha_{out} }{ 1-\alpha_{out} } \leq \max \left( \frac{ \alpha }{ 1-\alpha }, \frac{ \beta_1 }{ \beta_d } \right).
	\diagS^{out} \leq \max \left(  \diagS , \frac{ \beta_1 }{\beta_1+\beta_d } \right). 
	\end{equation}
\end{thm}

\begin{proof} 
	Without loss of generality, we assume that 
	the initial state is diagonal.	
	
	First, consider the situation where 
	\[\alpha \geq \frac{\beta_1}{\beta_1+\beta_{\dimHB}}
	\rightarrow  \alpha \beta_d \geq ( 1-\alpha) 
	 \beta_1.\]
	This indicates no crossing, i.e.  $m=0$ for which
	equation (\ref{eq:alphaout_crossing}) gives 
	$\diagS_{out} = \diagS$. 
	This is in agreement with equation (\ref{eq:thm2}). 
		
	For the case of $ \alpha \leq \frac{\beta_1}{\beta_1+\beta_{\dimHB}}$, 
	we start with pointing that  %, we start with 
	$\beta_j\leq \beta_1$
	and $\beta_{\dimHB}\leq \beta_{\dimHB-j+1}$ for all $j$. 
	Multiplying these two inequalities gives 
	\begin{equation}\label{eq:th2-midstep1}
	\beta_j \beta_{\dimHB}\leq \beta_1 \beta_{\dimHB-j+1}, \; \forall j.
	\end{equation}
	If we sum over the first $ m $ values of $ j $, we get 
	\[(\beta_1 + \cdots+\beta_{m}) \beta_d \leq \beta_1 (\beta_{d-m+1} + \cdots+\beta_{d}) \] 
	or
	\begin{align}\label{eq:th2-midstep2}
	& \delta_1 \beta_d  - \beta_1 \delta_3 \leq 0 \cr
	\Rightarrow & \beta_d \delta_1 - \beta_1 \delta_3 +\beta_1 \leq \beta_1 \cr
	\Rightarrow & \beta_d \delta_1 - \beta_1 \delta_3 +\beta_1 (\delta_1+\delta_2+\delta_3) = \cr
	%\Rightarrow & \beta_d \delta_1 - \beta_1 \delta_3 +\beta_1 = (\beta_1+\beta_d)\delta_1 + \beta_1\delta_2 \cr 
	 & (\beta_1+\beta_d)(\delta_1 + \frac{\beta_1}{\beta_1+\beta_d} \delta_2) \leq \beta_1 
	\end{align}
	Since $\alpha\leq \frac{\beta_1}{\beta_1+\beta_d}$, 
	we get 
	\begin{align}
	(\beta_1+\beta_d)(\delta_1 + \alpha \delta_2) 
	\leq (\beta_1+\beta_d)(\delta_1 + \frac{\beta_1}{\beta_1+\beta_d} \delta_2)\leq \beta_1\cr 
	\end{align}
	From equation (\ref{eq:alphaout_crossing}) of the main text, we get 
	\[\alpha_{out} \leq \frac{\beta_1}{\beta_1+\beta_d}, \]
	and this concludes the proof. 
\end{proof}

Similar to theorem (\ref{thm:no-go}), the result of theorem 
(\ref{thm:cooling_power}) can be rewritten as 
\begin{equation}
\epsilon\left( \rho^{out}_{T} \right) \leq \max 
\left( \epsilon\left( \rho_{T} \right), \varepsilon\left( \rho_{A} \right) \right),
\end{equation}
with $ \epsilon\left( \rho_{T} \right)$ and $ \varepsilon\left( \rho_{A} \right) $ the 
polarization of the target and auxiliary elements. 

This indicates that only the ratio of the largest 
to the smallest element of $ \lambda\left( \rho_A\right) $ affects the cooling bound. 
This is similar to the result from \cite{allahverdyan2011thermodynamic}.

This theorem extends the results of the first theorem to the situation where the reset element is a multi-level quantum system, i.e. a qudit. 
More specifically, it shows that no unitary operation can compress and transfer entropy for a system comprised of a qubit and qudit. 
Consider the example where the target is a qubit and the reset element is a qudit and initially, the qudit is more polarized i.e. $\epsilon_{\text{T}}\leq \varepsilon_{\text{R}}$. Theorem (2) shows that there is no unitary operation that can polarize the target qubit beyond $\varepsilon_{\text{R}}$. 

\subsubsection*{Unitarity limitations and HBAC limit}
%%% Connection/Interpretation of the Thm for the HBAC limit
These two theorems impose an upper bound for the 
recursive HBAC technique. 

To derive the upper bound, we assume that all the 
qubits are initially less polarized than the reset, i.e. 
$\epsilon(Q_i)\leq \epsilon_R$. 
For $ Q_2 $ the polarization that can be achieved with unitary compression is bounded by $\epsilon_R$, as indicated by theorem (\ref{thm:no-go}).
 
For $ Q_3 $, the combination of the two first qubits
serves as the auxiliary element. 
Theorem (\ref{thm:cooling_power}) indicates that 
$\epsilon(Q_3) \leq \varepsilon(\rho_{1,2}) = 2\epsilon_R$.
Next, the first three qubits make the auxiliary element for the fourth qubit with $\varepsilon(\rho_{1,2,3}) = 4\epsilon_R$. This also means that from each qubit to the next, the polarization bound doubles.
For $ Q_n $, the combination of the $n-1$ first 
qubits make the auxiliary element with 
$\varepsilon = 2^{n-2} \epsilon_R$.  
This gives a bound of  $ 2^{n-2} \epsilon_{R} $ for qubit $ Q_n $.
See figure (\ref{fig:HBAC_limit}).

The resulting bound coincides with the cooling limit 
of the Heat-Bath Algorithmic cooling techniques 
established in \cite{raeisi2015asymptotic, raeisi2019novel}. 
This also indicates that the algorithm described above presents a new method for HBAC that converges to the HBAC limit, although it is neither efficient nor practical \cite{raeisi2019novel}. 

But the interesting result is that 
any changes to the bounds in theorems 
(\ref{thm:no-go}) and (\ref{thm:cooling_power})  
would change the HBAC limit. 
To this end, assume that it was possible to find a 
unitary operation to compress the entropy and 
get $\epsilon_T^{out} = \gamma \max\left( \epsilon_T, 
\varepsilon_{A}\right)  $, with $ \gamma > 1 $ (in violation of the two theorems). 

Then if we follow the same steps as we did above 
the output polarization 
should be $ \epsilon_T^{out} = \gamma \epsilon_{R}$, 
which exceeds the limit of HBAC. 

Similarly, for the 
three-qubit HBAC, instead of the limit of $ 2\epsilon_R $, we would be able to increase the purity of $Q_3$ to $ \gamma(\gamma+1) \epsilon_R$ and for $ n $
qubits it would be possible to increase the polarization to
$ \gamma(\gamma+1)^{n-2}\epsilon_{0} $. This exceeds the limit of HBAC. 

More precisely, this shows that if there were a unitary operation that could violate the bounds in theorems (1) and (2), then it would have been possible to exceed the HBAC cooling limit.

%%% Connection/Interpretation of the Thm for the HBAC limit

\subsection*{Conclusion}

%%%Conclusion

In conclusion, we investigated the roots of the 
cooling limit of HBAC. 
We showed that unitary operations cannot compress entropy beyond the initial entropies of the target and reset elements. 
This means that using unitary compression for HBAC, it is not possible to increase the purity beyond the maximum of the individual purities. 
We proved this for both a qubit and a qudit reset. 
This means that the unitarity of the compression operation
imposes limits on the compression and we showed that these limits lead to the HBAC cooling limit. 
Specifically, we introduced a new HBAC algorithm and showed that without the limitations imposed by the unitarity of the compression operations, the new HBAC technique would exceed the 
limit of HBAC. But the restrictions of the unitarity lead exactly to the limit of HBAC. 
This shows that the root of the cooling 
limit of HBAC is in the unitarity of the compression operation. 

It is interesting to use these results to understand how non-unitary compression operations might help to improve HBAC beyond the current limit. 
In particular, with our algorithm, it is expected that 
if the unitary compression is replaced by a non-unitary operation where the restrictions of theorem (1) and (2) do not apply, the cooling limit of HBAC would no longer hold. It is also interesting to explore what families of completely positive and trace preserving (CPTP) maps can be practically used to improve beyond the current scope of HBAC and to understand how far the limit can be pushed. 

%%%Conclusion

%\subsection*{end}
\begin{acknowledgments}
	We would like to thank Michele Mosca for fruitful discussions. This work was supported by the research 
	grant system of Sharif University of Technology (G960219). %,
\end{acknowledgments}

\bibliographystyle{apsrev4-1}
\bibliography{Purification}

\begin{thebibliography}{15}%
\makeatletter
\providecommand \@ifxundefined [1]{%
 \@ifx{#1\undefined}
}%
\providecommand \@ifnum [1]{%
 \ifnum #1\expandafter \@firstoftwo
 \else \expandafter \@secondoftwo
 \fi
}%
\providecommand \@ifx [1]{%
 \ifx #1\expandafter \@firstoftwo
 \else \expandafter \@secondoftwo
 \fi
}%
\providecommand \natexlab [1]{#1}%
\providecommand \enquote  [1]{``#1''}%
\providecommand \bibnamefont  [1]{#1}%
\providecommand \bibfnamefont [1]{#1}%
\providecommand \citenamefont [1]{#1}%
\providecommand \href@noop [0]{\@secondoftwo}%
\providecommand \href [0]{\begingroup \@sanitize@url \@href}%
\providecommand \@href[1]{\@@startlink{#1}\@@href}%
\providecommand \@@href[1]{\endgroup#1\@@endlink}%
\providecommand \@sanitize@url [0]{\catcode `\\12\catcode `\$12\catcode
  `\&12\catcode `\#12\catcode `\^12\catcode `\_12\catcode `\%12\relax}%
\providecommand \@@startlink[1]{}%
\providecommand \@@endlink[0]{}%
\providecommand \url  [0]{\begingroup\@sanitize@url \@url }%
\providecommand \@url [1]{\endgroup\@href {#1}{\urlprefix }}%
\providecommand \urlprefix  [0]{URL }%
\providecommand \Eprint [0]{\href }%
\providecommand \doibase [0]{http://dx.doi.org/}%
\providecommand \selectlanguage [0]{\@gobble}%
\providecommand \bibinfo  [0]{\@secondoftwo}%
\providecommand \bibfield  [0]{\@secondoftwo}%
\providecommand \translation [1]{[#1]}%
\providecommand \BibitemOpen [0]{}%
\providecommand \bibitemStop [0]{}%
\providecommand \bibitemNoStop [0]{.\EOS\space}%
\providecommand \EOS [0]{\spacefactor3000\relax}%
\providecommand \BibitemShut  [1]{\csname bibitem#1\endcsname}%
\let\auto@bib@innerbib\@empty
%</preamble>
\bibitem [{\citenamefont {Wang}\ \emph {et~al.}(2013)\citenamefont {Wang},
  \citenamefont {Ghobadi}, \citenamefont {Raeisi},\ and\ \citenamefont
  {Simon}}]{wang2013precision}%
  \BibitemOpen
  \bibfield  {author} {\bibinfo {author} {\bibfnamefont {T.}~\bibnamefont
  {Wang}}, \bibinfo {author} {\bibfnamefont {R.}~\bibnamefont {Ghobadi}},
  \bibinfo {author} {\bibfnamefont {S.}~\bibnamefont {Raeisi}}, \ and\ \bibinfo
  {author} {\bibfnamefont {C.}~\bibnamefont {Simon}},\ }\href@noop {}
  {\bibfield  {journal} {\bibinfo  {journal} {Physical Review A}\ }\textbf
  {\bibinfo {volume} {88}},\ \bibinfo {pages} {062114} (\bibinfo {year}
  {2013})}\BibitemShut {NoStop}%
\bibitem [{\citenamefont {Nielsen}\ and\ \citenamefont
  {Chuang}(2002)}]{nielsen2002quantum}%
  \BibitemOpen
  \bibfield  {author} {\bibinfo {author} {\bibfnamefont {M.~A.}\ \bibnamefont
  {Nielsen}}\ and\ \bibinfo {author} {\bibfnamefont {I.}~\bibnamefont
  {Chuang}},\ }\href@noop {} {\enquote {\bibinfo {title} {Quantum computation
  and quantum information},}\ } (\bibinfo {year} {2002})\BibitemShut {NoStop}%
\bibitem [{\citenamefont {Ben-Or}\ \emph {et~al.}(2013)\citenamefont {Ben-Or},
  \citenamefont {Gottesman},\ and\ \citenamefont {Hassidim}}]{ben2013quantum}%
  \BibitemOpen
  \bibfield  {author} {\bibinfo {author} {\bibfnamefont {M.}~\bibnamefont
  {Ben-Or}}, \bibinfo {author} {\bibfnamefont {D.}~\bibnamefont {Gottesman}}, \
  and\ \bibinfo {author} {\bibfnamefont {A.}~\bibnamefont {Hassidim}},\
  }\href@noop {} {\bibfield  {journal} {\bibinfo  {journal} {arXiv preprint
  arXiv:1301.1995}\ } (\bibinfo {year} {2013})}\BibitemShut {NoStop}%
\bibitem [{\citenamefont {Shor}(1994)}]{shor1994algorithms}%
  \BibitemOpen
  \bibfield  {author} {\bibinfo {author} {\bibfnamefont {P.~W.}\ \bibnamefont
  {Shor}},\ }in\ \href@noop {} {\emph {\bibinfo {booktitle} {Proceedings 35th
  annual symposium on foundations of computer science}}}\ (\bibinfo
  {organization} {Ieee},\ \bibinfo {year} {1994})\ pp.\ \bibinfo {pages}
  {124--134}\BibitemShut {NoStop}%
\bibitem [{\citenamefont {Schulman}\ and\ \citenamefont
  {Vazirani}(1999)}]{schulman_scalable_1998}%
  \BibitemOpen
  \bibfield  {author} {\bibinfo {author} {\bibfnamefont {L.~J.}\ \bibnamefont
  {Schulman}}\ and\ \bibinfo {author} {\bibfnamefont {U.~V.}\ \bibnamefont
  {Vazirani}},\ }in\ \href {\doibase 10.1145/301250.301332} {\emph {\bibinfo
  {booktitle} {Proceedings of the Thirty-first Annual ACM Symposium on Theory
  of Computing}}},\ \bibinfo {series and number} {STOC '99}\ (\bibinfo
  {publisher} {ACM},\ \bibinfo {address} {New York, NY, USA},\ \bibinfo {year}
  {1999})\ pp.\ \bibinfo {pages} {322--329}\BibitemShut {NoStop}%
\bibitem [{\citenamefont {Boykin}\ \emph {et~al.}(2002)\citenamefont {Boykin},
  \citenamefont {Mor}, \citenamefont {Roychowdhury}, \citenamefont {Vatan},\
  and\ \citenamefont {Vrijen}}]{boykin_algorithmic_2002}%
  \BibitemOpen
  \bibfield  {author} {\bibinfo {author} {\bibfnamefont {P.~O.}\ \bibnamefont
  {Boykin}}, \bibinfo {author} {\bibfnamefont {T.}~\bibnamefont {Mor}},
  \bibinfo {author} {\bibfnamefont {V.}~\bibnamefont {Roychowdhury}}, \bibinfo
  {author} {\bibfnamefont {F.}~\bibnamefont {Vatan}}, \ and\ \bibinfo {author}
  {\bibfnamefont {R.}~\bibnamefont {Vrijen}},\ }\href {\doibase
  10.1073/pnas.241641898} {\bibfield  {journal} {\bibinfo  {journal}
  {Proceedings of the National Academy of Sciences}\ }\textbf {\bibinfo
  {volume} {99}},\ \bibinfo {pages} {3388} (\bibinfo {year} {2002})},\ \bibinfo
  {note} {{PMID:} 11904402}\BibitemShut {NoStop}%
\bibitem [{\citenamefont {Zaiser}\ \emph {et~al.}(2021)\citenamefont {Zaiser},
  \citenamefont {Masatth}, \citenamefont {Rao}, \citenamefont {Raeisi},\ and\
  \citenamefont {Wrachtrup}}]{zaiser2021experimental}%
  \BibitemOpen
  \bibfield  {author} {\bibinfo {author} {\bibfnamefont {S.}~\bibnamefont
  {Zaiser}}, \bibinfo {author} {\bibfnamefont {B.}~\bibnamefont {Masatth}},
  \bibinfo {author} {\bibfnamefont {D.}~\bibnamefont {Rao}}, \bibinfo {author}
  {\bibfnamefont {S.}~\bibnamefont {Raeisi}}, \ and\ \bibinfo {author}
  {\bibfnamefont {J.}~\bibnamefont {Wrachtrup}},\ }\href {\doibase
  10.1038/s41534-021-00408-z} {\bibfield  {journal} {\bibinfo  {journal} {npj
  Quantum Information}\ }\textbf {\bibinfo {volume} {7}},\ \bibinfo {pages}
  {92} (\bibinfo {year} {2021})}\BibitemShut {NoStop}%
\bibitem [{\citenamefont {Schulman}\ \emph {et~al.}(2005)\citenamefont
  {Schulman}, \citenamefont {Mor},\ and\ \citenamefont
  {Weinstein}}]{schulman_physical_2005}%
  \BibitemOpen
  \bibfield  {author} {\bibinfo {author} {\bibfnamefont {L.~J.}\ \bibnamefont
  {Schulman}}, \bibinfo {author} {\bibfnamefont {T.}~\bibnamefont {Mor}}, \
  and\ \bibinfo {author} {\bibfnamefont {Y.}~\bibnamefont {Weinstein}},\ }\href
  {\doibase 10.1103/PhysRevLett.94.120501} {\bibfield  {journal} {\bibinfo
  {journal} {Phys. Rev. Lett.}\ }\textbf {\bibinfo {volume} {94}},\ \bibinfo
  {pages} {120501} (\bibinfo {year} {2005})}\BibitemShut {NoStop}%
\bibitem [{\citenamefont {Raeisi}\ and\ \citenamefont
  {Mosca}(2015)}]{raeisi2015asymptotic}%
  \BibitemOpen
  \bibfield  {author} {\bibinfo {author} {\bibfnamefont {S.}~\bibnamefont
  {Raeisi}}\ and\ \bibinfo {author} {\bibfnamefont {M.}~\bibnamefont {Mosca}},\
  }\href@noop {} {\bibfield  {journal} {\bibinfo  {journal} {Physical review
  letters}\ }\textbf {\bibinfo {volume} {114}},\ \bibinfo {pages} {100404}
  (\bibinfo {year} {2015})}\BibitemShut {NoStop}%
\bibitem [{\citenamefont {Browne}\ \emph {et~al.}(2014)\citenamefont {Browne},
  \citenamefont {Garner}, \citenamefont {Dahlsten},\ and\ \citenamefont
  {Vedral}}]{browne2014guaranteed}%
  \BibitemOpen
  \bibfield  {author} {\bibinfo {author} {\bibfnamefont {C.}~\bibnamefont
  {Browne}}, \bibinfo {author} {\bibfnamefont {A.~J.}\ \bibnamefont {Garner}},
  \bibinfo {author} {\bibfnamefont {O.~C.}\ \bibnamefont {Dahlsten}}, \ and\
  \bibinfo {author} {\bibfnamefont {V.}~\bibnamefont {Vedral}},\ }\href@noop {}
  {\bibfield  {journal} {\bibinfo  {journal} {Physical review letters}\
  }\textbf {\bibinfo {volume} {113}},\ \bibinfo {pages} {100603} (\bibinfo
  {year} {2014})}\BibitemShut {NoStop}%
\bibitem [{\citenamefont {Reeb}\ and\ \citenamefont
  {Wolf}(2014)}]{reeb2014improved}%
  \BibitemOpen
  \bibfield  {author} {\bibinfo {author} {\bibfnamefont {D.}~\bibnamefont
  {Reeb}}\ and\ \bibinfo {author} {\bibfnamefont {M.~M.}\ \bibnamefont
  {Wolf}},\ }\href@noop {} {\bibfield  {journal} {\bibinfo  {journal} {New
  Journal of Physics}\ }\textbf {\bibinfo {volume} {16}},\ \bibinfo {pages}
  {103011} (\bibinfo {year} {2014})}\BibitemShut {NoStop}%
\bibitem [{\citenamefont {Masanes}\ and\ \citenamefont
  {Oppenheim}(2017)}]{masanes2017general}%
  \BibitemOpen
  \bibfield  {author} {\bibinfo {author} {\bibfnamefont {L.}~\bibnamefont
  {Masanes}}\ and\ \bibinfo {author} {\bibfnamefont {J.}~\bibnamefont
  {Oppenheim}},\ }\href@noop {} {\bibfield  {journal} {\bibinfo  {journal}
  {Nature communications}\ }\textbf {\bibinfo {volume} {8}},\ \bibinfo {pages}
  {14538} (\bibinfo {year} {2017})}\BibitemShut {NoStop}%
\bibitem [{\citenamefont {Gour}\ \emph {et~al.}(2015)\citenamefont {Gour},
  \citenamefont {M{\"u}ller}, \citenamefont {Narasimhachar}, \citenamefont
  {Spekkens},\ and\ \citenamefont {Halpern}}]{gour2015resource}%
  \BibitemOpen
  \bibfield  {author} {\bibinfo {author} {\bibfnamefont {G.}~\bibnamefont
  {Gour}}, \bibinfo {author} {\bibfnamefont {M.~P.}\ \bibnamefont
  {M{\"u}ller}}, \bibinfo {author} {\bibfnamefont {V.}~\bibnamefont
  {Narasimhachar}}, \bibinfo {author} {\bibfnamefont {R.~W.}\ \bibnamefont
  {Spekkens}}, \ and\ \bibinfo {author} {\bibfnamefont {N.~Y.}\ \bibnamefont
  {Halpern}},\ }\href@noop {} {\bibfield  {journal} {\bibinfo  {journal}
  {Physics Reports}\ }\textbf {\bibinfo {volume} {583}},\ \bibinfo {pages} {1}
  (\bibinfo {year} {2015})}\BibitemShut {NoStop}%
\bibitem [{\citenamefont {Allahverdyan}\ \emph {et~al.}(2011)\citenamefont
  {Allahverdyan}, \citenamefont {Hovhannisyan}, \citenamefont {Janzing},\ and\
  \citenamefont {Mahler}}]{allahverdyan2011thermodynamic}%
  \BibitemOpen
  \bibfield  {author} {\bibinfo {author} {\bibfnamefont {A.~E.}\ \bibnamefont
  {Allahverdyan}}, \bibinfo {author} {\bibfnamefont {K.~V.}\ \bibnamefont
  {Hovhannisyan}}, \bibinfo {author} {\bibfnamefont {D.}~\bibnamefont
  {Janzing}}, \ and\ \bibinfo {author} {\bibfnamefont {G.}~\bibnamefont
  {Mahler}},\ }\href@noop {} {\bibfield  {journal} {\bibinfo  {journal}
  {Physical Review E}\ }\textbf {\bibinfo {volume} {84}},\ \bibinfo {pages}
  {041109} (\bibinfo {year} {2011})}\BibitemShut {NoStop}%
\bibitem [{\citenamefont {Raeisi}\ \emph {et~al.}(2019)\citenamefont {Raeisi},
  \citenamefont {Kieferov{\'a}},\ and\ \citenamefont
  {Mosca}}]{raeisi2019novel}%
  \BibitemOpen
  \bibfield  {author} {\bibinfo {author} {\bibfnamefont {S.}~\bibnamefont
  {Raeisi}}, \bibinfo {author} {\bibfnamefont {M.}~\bibnamefont
  {Kieferov{\'a}}}, \ and\ \bibinfo {author} {\bibfnamefont {M.}~\bibnamefont
  {Mosca}},\ }\href@noop {} {\bibfield  {journal} {\bibinfo  {journal}
  {Physical Review Letters}\ }\textbf {\bibinfo {volume} {122}},\ \bibinfo
  {pages} {220501} (\bibinfo {year} {2019})}\BibitemShut {NoStop}%
\end{thebibliography}%

\appendix

\section{Lower bounds for compression}
It is not easy to lower bound the output polarization
in terms of the initial polarization of the target 
and the auxiliary element. 

As a simple example, let's consider the two qubit case. 
Assume that we are given two qubits 
for the target and the auxiliary systems with 
$\bar{\lambda}\left( \rho_{\Sys} \right) = \{\alpha, 1-\alpha\}$ and 
$\bar{\lambda}\left( \rho_{\A} \right) = \{\beta, 1-\beta\}$ and that the state of the target element  
after compression is given by 
$\bar{\lambda}\left( \rho_{\Sys}^{out} \right) = 
\{\alpha^{out}, 1-\alpha^{out}\}$.

The initial state of the full system is 
\[\lambda\left( \rho_{T,A}\right) =   \{\alpha\beta,\alpha(1-\beta), (1-\alpha)\beta, 
(1-\alpha)(1-\beta)\}.\]
Now consider a unitary compression operation that takes this to 
\begin{equation}\label{sm_eq:mixer}
\lambda\left( \rho^{out}_{T,A}\right) =   \{\alpha\beta,(1-\alpha)(1-\beta),\alpha(1-\beta), (1-\alpha)\beta\}.
\end{equation}
It is easy to see that 
$ \alpha^{out} =  \alpha\beta + (1-\alpha)(1-\beta)$ and 
it can be smaller than the minimum of $ \alpha $ and $ \beta $. 
Figure (\ref{fig:LowerBound}) shows $ \alpha^{out} - \min\left(  \alpha , \beta\right)  $ 
and it is clear that it is mostly negative. 

\begin{figure}
	\begin{centering}
		\includegraphics[width=\columnwidth]{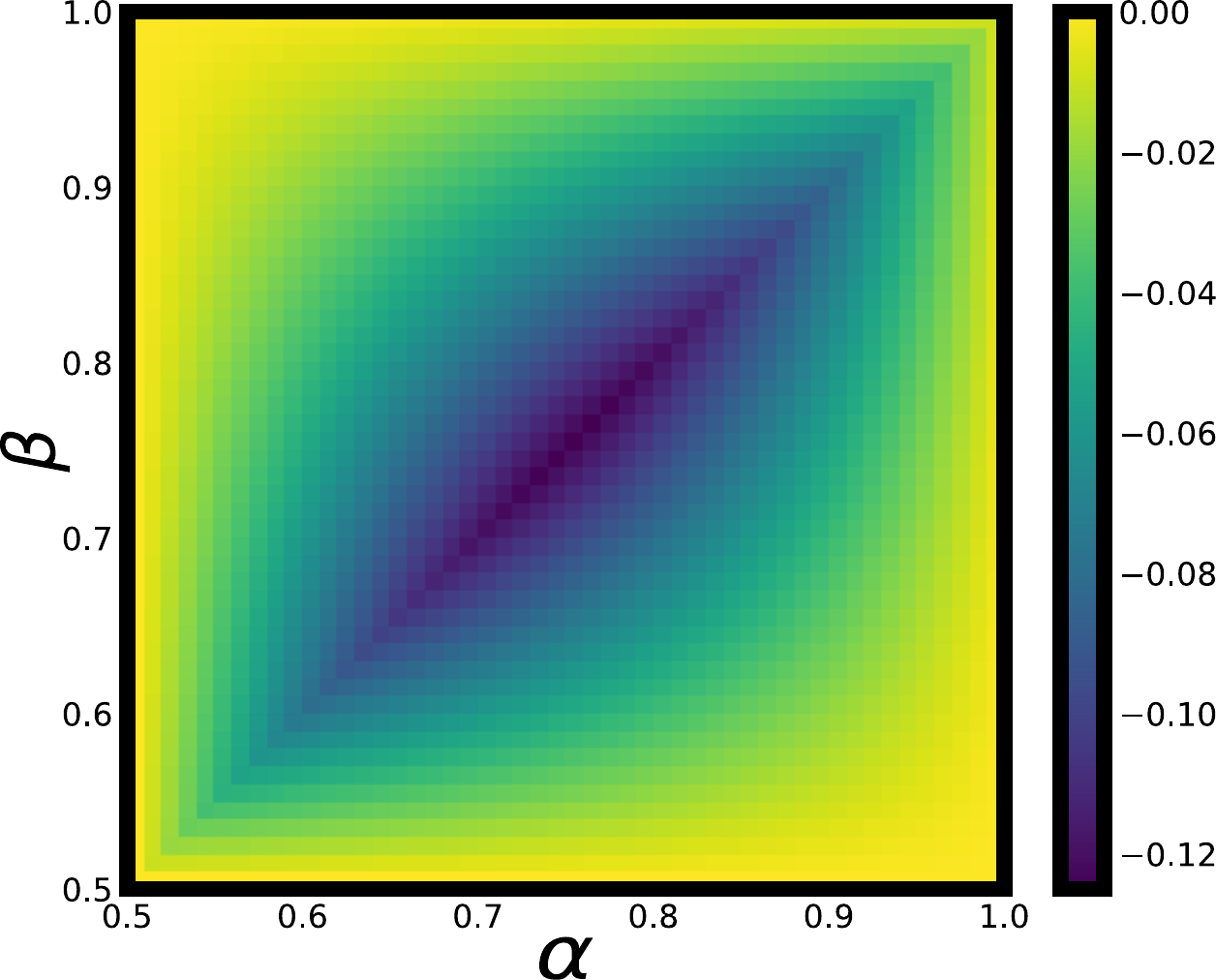}
		\par\end{centering}
	
	\caption{\label{fig:LowerBound}
		Lower-bound of purification: 
		It may seem 
		that the output polarization should be 
		lower-bounded  by the minimum of the initial 
		polarization of the target and reset qubits. 
		This plot shows that this is not true. It gives a 
		specific example of two qubit purification 
		with the purification operator explained 
		in the equation (\ref{sm_eq:mixer}). 
		On the axes are the largest eigenvalue, 
		$ \alpha , \beta $
		of the initial density matrix of the target
		and the auxiliary elements. The color-bar 
		shows the $ \alpha^{out} - 
		\min\left(  \alpha , \beta\right)  $ . 
		The negativity of the plot shows that 
		the output polarization 
		is, for the most part,  less than both the initial target 
		and auxiliary element. 
	}
\end{figure}

\section{Tightness of the bounds}
Next we discuss the tightness of the bounds. 

For the two qubits, the bound in theorem  (1) is tight. If the polarization of the auxiliary element is 
greater than the target, the states can be swapped and otherwise, the target is already at the bound of the purification. 

The bound in theorem  (2), 
is also tight. In figure  
(1-b) of the main text, 
the points on $ y=1 $ represent instances where the output polarization is equal to the maximum of the initial polarizations of the target and auxiliary elements. 

However, for theorem  (2) the bound is tight only when the initial polarization 
of the target is greater than or equal to the auxiliary element. 
However, if $  \epsilon\left( \rho_{T} \right) 
< \varepsilon\left( \rho_{A} \right) $, 
then we can back-track the steps of the proof and show that 
the output polarization is always less than the 
polarization of the auxiliary element, i.e. 
$ \epsilon\left( \rho^{out}_{T} \right) <
\varepsilon\left( \rho_{A} \right) $. For this
we can go back to the equation (11), 
where there is at least one value of 
$ j $ for which the inequalities are strict. 
Otherwise, all the $ \beta_j $ should be equal which gives a maximally mixed state for the auxiliary element and therefore it cannot have higher polarization than the target element. 
This means that the inequality in equation (12) of the main text and the final result should also be strict and 
as a result, for the situation where the 
auxiliary element is initially more polarized, 
the bound is no longer tight. 
It is also confirmed that all the points in figure (1-B) that saturate the limit ($ y=1 $), are cases where 
the target element is initially more polarized. 

For the open-system setting, 
Raeisi and Mosca showed in 
\cite{raeisi2015asymptotic} 
that HBAC asymptotically converges to 
this limit which indicates that the 
bound is asymptotically tight and  
it is possible to get arbitrarily close 
to the bound. For the proof of convergence 
see \cite{raeisi2015asymptotic} . 

\end{document}